\documentclass[conference, twocolumn]{IEEEtran}
\usepackage{amsmath,amssymb,graphicx,epsfig,fancyhdr,amsthm,tabulary,epstopdf}
\usepackage{tcolorbox}

\usepackage{hyperref}
\usepackage{cite}
\usepackage{multirow}
\usepackage{float}
\usepackage{amsfonts}
\usepackage{footnote}
\makesavenoteenv{tabular}
\usepackage{tabulary,bigstrut,array,multirow,booktabs,subfigure}%
\usepackage{rotating}

\usepackage{tikz}
\usetikzlibrary{calc}
\usepackage{pdfsync}
\usepackage{bm}

\makeatletter
\newcommand*\dashline{\rotatebox[origin=c]{90}{$\dabar@\dabar@\dabar@\dabar@\dabar@\dabar@\dabar@$}}

\newtcolorbox{mybox}{colback=red!5!white,colframe=red!75!black}

\makeatletter
\newcommand*{\rom}[1]{\expandafter\@slowromancap\romannumeral #1@}
\makeatother

\newtheoremstyle{cited}%
  {3pt}% (space above)
  {3pt}% (space below)
  {\itshape}% (body font)
  {}% (indent amount)
  {\bfseries}% {theorem head font}
  {{\fontsize{15}{15}\selectfont .}}% {punctuation after theorem head}
  {.5em}% {space after theorem head}
  {\thmname{#1} \thmnumber{#2} \thmnote{\normalfont#3}}% {theorem head spec}

\newtheorem{lem}{Lemma}

\theoremstyle{remark}
\newtheorem{rem}{Remark}

\theoremstyle{definition}

\begin{document}
\title{Wireless Powered Cooperative Relaying using NOMA with Imperfect CSI}
\author{Dinh-Thuan Do\textsuperscript{*}, Mojtaba Vaezi\textsuperscript{\dag},   and Thanh-Luan Nguyen{\ddag}\\ \vspace{0mm}
   \IEEEauthorblockN{}\vspace{0.5mm}
    \IEEEauthorblockA{
         \textsuperscript{*}Wireless Communications Research Group, Faculty of Electrical and Electronics Engineering, Ton Duc \\ Thang University, Ho Chi Minh City, Vietnam\\
         \textsuperscript{\dag}Department of Electrical and Computer Engineering, Villanova University, PA 19085, USA\\
        \textsuperscript{\ddag}Department of Electrical Engineering, Bach Khoa University, Vietnam\\
    Email: dodinhthuan@tdt.edu.vn, mvaezi@villanova.edu, luannguyen.cce@gmail.com
   }
}

\maketitle

\begin{abstract}
The impact of imperfect channel state (CSI) information in an energy harvesting (EH)
 cooperative non-orthogonal multiple
access (NOMA) network, consisting of a source,
 two users, and an EH relay is investigated in this paper.
 The relay is not equipped with a fixed
power source and acts as a wireless powered node to help signal
transmission to the users.
Closed-form expressions for the outage probability of both users are derived  under
imperfect CSI
%and Nakagami-$m$ fading
for two different power allocation strategies namely fixed and dynamic power allocation.
Monte Carlo simulations are used to numerically evaluate the effect of imperfect CSI.
These results  confirm the theoretical outage analysis  and
show  that NOMA can outperform orthogonal multiple access even with imperfect CSI.
\end{abstract}


\section{Introduction}
\label{sec:intro}

Effective radio access technology development
plays a key role in achieving  high data rate and enabling
massive connectivity for  next-generation  wireless networks. Non-orthogonal multiple access (NOMA)
is one of the key enabling technologies toward this goal \cite{NOMAbook,saito2013non, ding2014performance,shin2017non, yeu2018exploiting}.
NOMA improves the spectral efficiency and increases the number of connections
and is seen as a promising candidate in 5G and beyond cellular networks.

Energy harvesting (EH),  on the other hand,
is an important process toward green communication systems
in which wasted or unimportant
energy, e.g., sound and  radio frequency signals, are captured and converted into electricity.
By allowing wireless nodes to recharge their batteries from the radio frequency
signals, rather than  the traditional energy sources,
EH can play  an important role in powering wireless devices and has the
ability to prolong the  lifetime of the  energy-constrained nodes.
Wireless  power transfer is one of the EH technologies in which
green energy can be harvested using either  the ambient signals or
 dedicated  power sources,
 and has recently emerged as a potential technology to assist
  reduction of power consumption \cite{pin2013ambient,val2014har,sbi2015wirel}.
Simultaneous wireless information and power transfer (SWIPT)
is a technology that enables transmission of
both power and information at the same time.
SWIPT can improve power
consumption and spectral efficiency and also can
be used for interference management.
Understanding  the theoretical aspects and potentials of SWIPT  has attracted much research interest  \cite{rzhang2013mimo,lliu2013wirel,jpark2013joint}.

Motivated by the diverse requirements  of 5G systems \cite{ITUSG05},
combining  NOMA and SWIPT has
recently been studied by different groups in various settings
\cite{liu2016cooperative,han2016performance, nguyen2017maximum, do2017bnbf,yang2017impact}.
Specifically, the application of SWIPT to NOMA  in a single-input-single-output system is
explored in  \cite{liu2016cooperative}, in which
three user-paring schemes are examined via their outage performance.
 The outage performance for EH-enabled NOMA relaying networks is  studied in
  \cite{han2016performance}. Several NOMA user selection models, including random  user
selection and  distance-based  user
selections  are  investigated in \cite{do2017bnbf}, where
the best-near best-far  user selection scheme is promoted to enhance the outage performance.
In the context of cooperative systems, SWIPT-NOMA  is introduced to lengthen the lifetime
 of the energy-constrained relay nodes.  Joint power allocation  and
 user selection is then considered to minimize the outage probability  of
all  users \cite{yang2017impact}.
Most of the prior works have, however, assumed perfect
channel state information (CSI) and  Rayleigh fading
channels. Due to  the existence of the channel estimation errors,
such an assumption is too idealistic and
does  not provide an insight on the  performance of practical networks \cite{vaezi2018myth}.
The analysis of NOMA networks with imperfect CSI is particularly important since the encoding and decoding
 of NOMA highly depends on the quality of CSI.

\begin{figure}
\begin{center}
\includegraphics[scale=0.75]{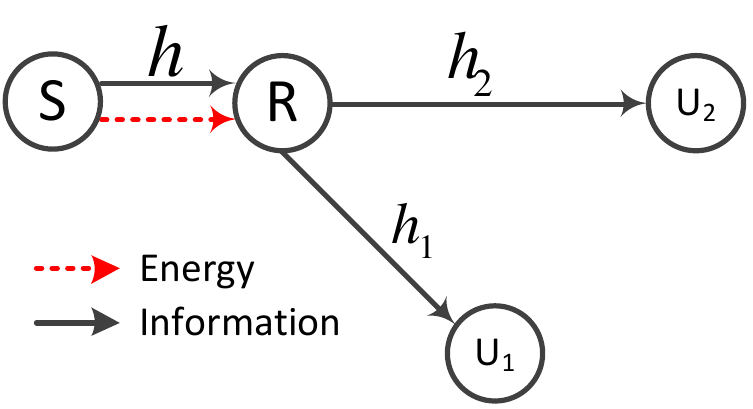}
\caption{A cooperative EH-NOMA consisting of a source (S)  with two users ($\rm{U}_1$ and  $\rm{U}_2$) simultaneously through an intermediate EH relay (R). Both S and R use NOMA for transmission.}
\label{fig:sys}
\end{center}
\end{figure}

In this paper, we investigate the effect of imperfect CSI in an EH-NOMA network. Specifically, we consider a cooperative NOMA network, consisting of a source,
two users, and an EH relay, as shown in Fig.~\ref{fig:sys},  in which the relay simultaneously receives information and harvests energy so that it can use the harvested energy to help the source to transmit information to the users.
We assume that the perfect CSI of source-relay link can be obtained whereas an imperfect CSI of the relay-user links can be obtained.
We drive  closed-form expressions for the outage probabilities of the NOMA users under the above assumptions and for two different power allocation schemes, namely fixed and dynamic cases.
The obtained results obviously include the Rayleigh fading case by setting  $m = 1$.
The outage probabilities of the imperfect CSI case are then compared with those of
the perfect CSI to evaluate and better understand the effect of imperfect CSI.

The paper is organized as follows.  We describe
the system model in Section~\ref{sec:model}, and  evaluate  outage probability of the users
under  fixed power allocation and  dynamic power allocation in  Section~\ref{sec:fpa} and
Section~\ref{sec:DPA}, respectively. The effect of imperfect successive interference cancellation
is studied in  Section~\ref{sec:imp}.
Simulation results are presented in Section~\ref{sec:sim} and the
paper is concluded in Section~\ref{sec:con}.

Throughout this paper, $\mathbb{P} \{\cdot\}$ is the probability measure, and
$f_{{Z}} \left( \cdot \right) $ and $F_{{Z}} \left( \cdot \right) $ denotes the probability density function (PDF) and the cumulative
distribution function (CDF) of random variable $Z, $ respectively.
%In addition, the "$\rm{F}$" quotation stands for "Fixed" or FPA-NOMA while "$\rm{D}$" stands for "Dynamic" or VPA-NOMA.

\section{System Model}
\label{sec:model}

Consider the  cooperative NOMA network depicted in Fig.~\ref{fig:sys}.  consisting of a source  communicating with two users   through an intermediate EH relay. The relay uses decode-and-forward relaying to assist the transmission between the source and the two users. Both the source and the relay use NOMA to transmit the users signals. All nodes are equipped with a single antenna and operate in half-duplex mode \cite{do2018improving,yeu2018exploiting}.
We assume that perfect CSI of the source-relay link is available whereas the estimation of the relay-user links is imperfect.
This is a step towards a more practical setting in which CSI at the relay is obtained via channel estimation and thus is subject to errors, unlike most of the previous works in EH-NOMA.
All channels experience independent and identically distributed Rayleigh fading with zero mean and unit variance; hence, the channel coefficients $h$ and $h_n$, $ n \in\{ 1,2\}$, are modeled as complex Gaussian random variables ${\mathcal{CN}}(0,1)$.
\footnote{For the fixed power allocation in Section~\ref{sec:fpa},
the source-relay link is assumed to experience  Nakagami-$m$ fading which  is more general than  Rayleigh fading.}
We also define ${\tau_1} = {2^{2{R_1}}} - 1$, ${\tau_2} = {2^{2{R_2}}} - 1$ and ${\tau_0} = {2^{2{R_1} + 2{R_2}}} - 1$, where ${R_1}$  and ${R_2}$ (bits/s/Hz) are the target data rates for user~1 and user~2, respectively.

The transmission from the source to users consist of two phases. During the first phase, the source transmits a superimposed signal, $x_s \triangleq \sqrt{\xi}{s_2} + \sqrt{1 - \xi}{s_1}$, to both users through the relay, where $\xi \in \left[ {0,1} \right]$ denotes the power allocation at the source.
 	Let us define $P \triangleq \frac{P_s}{\sigma^2}$ and $P' \triangleq \frac{P_r}{\sigma^2}$
 in which  ${P_s}$ and ${P_r}$ is the transmit power of the source and the relay nodes, respectively, and  $\sigma^2$ is the variance of the additive white Gaussian noise  at the relay and  users. Similar to \cite{yang2017impact},  maximum transmit power to ensure signal decoding at the relay  is given as
\begin{equation}
\label{eq1}
P' = \frac{P_r}{\sigma^2} = {\eta}(Pg-\tau_0),\quad g > \frac{\tau_0}{P}
\end{equation}
 in which $\eta  \in \left[ {0,1} \right]$ denotes the energy harvesting efficiency,
$g \triangleq |h|^2d^{-\alpha }$ where $d$ is the distance between the source and the relay, and $\alpha$ is the path loss exponent. Similar to the source, the relay also adopts NOMA for transmission.

During the second phase, the relay transmits a superposition of the decoded signals, $x_r = \sqrt\delta{s_2}+\sqrt{1-\delta}{s_1}$, to both users, where $\delta  \in \left[ {0,1} \right]$ is the power allocation coefficient at the relay.

Due to imperfect CSI, the estimated channel gains of the relay-user links are given as $\hat{h}_n = {h_n}{d_n^{-\alpha/2}} + e_n$ ($n = 1, 2$), where $d_n$ is the distance between the relay and user $n$, $e_n \sim \mathcal{CN}(0,\sigma _e^2)$ is the channel estimation error and $\hat h_n$ is the estimated channel coefficients. We assume $e_n$ is independent of ${h_n}$, thus $\hat{h}_n$ can be modeled as a complex Gaussian random variable with zero mean and variance  $ d_n^{-\alpha } + \sigma_e^2$. Now, without loss of generality, let us order the estimated channel  as $|\hat{h}_1|^2 > |\hat{h}_2|^2 $. Thus, the instantaneous signal-to-interference-plus-noise-ratio (SINR) of  user~2 to detect its own signal is obtained as
\begin{equation}
\label{eq2}
{\gamma_2} = \frac{{\delta P'g_2}}{{\left( {1 - \delta } \right)P'g_2 + P'\sigma_e^2 + 1}},
\end{equation}
in which ${g_n} \triangleq {|{{\hat h_n}} |^2}$. In NOMA, optimal decoding order
depends on the channel gains \cite{vaezi2018myth}; the stronger user (user~1)  first decodes the weaker user's (user~2's) signal and next its own signal after successive interference cancellation (SIC) of user~2's signal. The instantaneous SINR  to detect $s_2$ at user~1 is given by
\begin{equation}
\label{eq3}
{\gamma_{{1 \mathord{\left/
 {\vphantom {1 2}} \right.
 \kern-\nulldelimiterspace} 2}}} = \frac{{\delta P'g_1}}{{\left( {1 - \delta } \right)P'g_1 + \sigma_e^2P' + 1}}.
\end{equation}

If ${\gamma_{{\rm{1/2}}}} \ge {\tau_2}$, user~1 can decode ${s_{\rm{2}}}$ and remove it from the received signal using SIC. Therefore, the instantaneous signal-to-noise ratio (SNR) at user~1 to decode $s_1$
after SIC\footnote{Due to imperfect CSI, there will be another error term
in the denominator of \eqref{eq3}, i.e.,  SIC cannot be perfect.
Equivalently, this can be  seen as a reduction in the target rate $\tau_1$.
This will be discussed in Section~\ref{sec:imp}.
} is given by
\begin{equation}
\label{eq4}
{\gamma_{\rm{1}}} = \frac{{\left( {1 - \delta } \right)P'g_1}}{{P'\sigma_e^2 + 1}}.
\end{equation}

\section{Fixed Power Allocation}
\label{sec:fpa}

In this section, we consider a fixed power allocation scenario for both users. Since the source node may have line-of-sight with the relay, the source-relay link is assumed to experience Nakagami-$m$ fading,  a more general case of Rayleigh fading in which the shape factor is $m = 1$. Consequently, the outage probability at user~1, $p_{\rm{F}}^{(1)}$ is expressed as
\begin{equation}
\label{eq5}
p_{\rm{F}}^{( 1 )} = 1 - \mathbb{P} \left\{ {g > \frac{\tau_0}{P},{\gamma_{1/2}} > {\tau_2},{{\gamma}_{\rm{1}}} > {\tau_1}} \right\}.
\end{equation}

\begin{lem} \label{lem1}
For the FPA case, the outage probability of user~1 when $\delta  >  \frac{{{\tau_2}}}{{{\tau_2} + 1}}$
can be expressed as
\begin{align}
\label{eq6}
p_{\rm{F}}^{( 1 )} =
1 - \sum\limits_{i = 0}^{m - 1} {\begin{gathered}{{m-1}\choose{i}}\end{gathered}e^{ - \frac{{{\tau_0}m}}{{d^{ - \alpha }P}}}{{\left( {\frac{{{\tau_0}m}}{{d^{ - \alpha }P}}} \right)}^{m - i - 1}}} \notag \\
 \times  \sum\limits_{k = 1}^3 {\frac{{{{\left( { - 1} \right)}^{k - 1}}}}{{{2^i}\Gamma \left( m \right)}}e^{ - \frac{{{{{\Phi }}_{{\rm{n1}}}}}}{{{\Omega_k}}}}\beta_{n,k}^{i + 1}{K_{i + 1}}\left( {{\beta_{n,k}}} \right)},
\end{align}
in which  ${K_v}\left( z \right)$ is the ${v^{{\rm{th}}}}$ order modified Bessel function of the second kind \cite{jeffrey2007table}, ${{\beta_{n,k}} = 2\sqrt {\frac{m}{{d^{ - \alpha }P}}\frac{{{{{\Phi }}_{{\rm{n2}}}}}}{{{\Omega_k}}}} }$ in which
\begin{align}\label{eq:Phi}
 \begin{bmatrix}
 	\Phi_{11}& \Phi_{12}\\
 	\Phi_{21}& \Phi_{22}
 \end{bmatrix}
 \triangleq  {\left[ {\frac{{{\tau_2}}}{{\delta \left( {{\tau_2} + 1} \right) - {\tau_2}}}\;\; \frac{{{\tau_1}}}{{1 - \delta }}} \right]^T}\left[ {\sigma_e^2 \;\; \frac{1}{\eta }} \right],
 \end{align}
${\Omega_1} = d_2^{ - \alpha } + \sigma_e^2$, ${\Omega_2} = \frac{{{\Omega_1}{\Omega_3}}}{{{\Omega_1} + {\Omega_3}}}$ and ${\Omega_3} = d_1^{ - \alpha } + \sigma_e^2$.
Also, note that $p_F^{\left( 1 \right)} = 1$  when $\delta  \le \frac{{{\tau_2}}}{{{\tau_2} + 1}}$.\end{lem}
\begin{proof} See Appendix A.
\end{proof}

Similar to (\ref{eq5}), we have $p_\text{F}^{(2)}=1-\mathbb{P} \left \{P' >0,\gamma_2 > \tau_2 \right\}$ as the outage probability for user 2. For $\delta  >  \frac{{{\tau_2}}}{{{\tau_2} + 1}}$, $p_\text{F}^{(2)}$ can be expressed as
\begin{align}
\label{eq8}
p_{\rm{F}}^{\left( 2 \right)} =
1 - \sum\limits_{i = 0}^{m - 1} {\begin{gathered}{{m-1}\choose{i}}\end{gathered}e^{ - \frac{{{\tau_0}m}}{{d^{ - \alpha }P}}}{{\left( {\frac{{{\tau_0}m}}{{d^{ - \alpha }P}}} \right)}^{m - i - 1}}} \notag \\
 \times  \frac{1}{{{2^i}\Gamma \left( m \right)}}e^{ { - \frac{{{{{\Phi }}_{11}}}}{{{\Omega_2}}}}}\beta_{1,2}^{i + 1}{K_{i + 1}}\left( {{\beta_{1,2}}} \right).
\end{align}
Note that $p_\text{F}^{\left(1 \right)} = p_\text{F}^{\left(2 \right)} = 1$ when $\delta  \le \frac{{{\tau_2}}}{{{\tau_2} + 1}}$.

\begin{rem}
When the source-relay link experiences a Rayleigh fading, we can substitute $m = 1$ into (\ref{eq6}) and (\ref{eq8}) to achieve the outage probability for both users.
\end{rem}

\begin{rem}
Diversity order is zero at both users; i.e.,
\begin{align}\label{eq:dof}
D_\text{F}^{(n)} =  - \lim _{P \to \infty } \frac{\log p_{\text{F} }^{(n)}}{\log P} =0, \quad n =1, 2.
\end{align}
This is because in the high SNR regime, i.e., when  $P \to  \infty$, we have $\beta _{n,k}^{i + 1}{K_{i + 1}}({\beta _{n,k}}) \to {2^i}\Gamma (i + 1)$ \cite{jeffrey2007table} and ${e^{ - \frac{{{\tau _0}m}}{{{d^{ - \alpha }}P}}}} \to 1$. Then,
$p_{\text{F}}^{(1)} = 1 - \sum\limits_{k = 1}^3 {{{( - 1)}^{k - 1}}{e^{ - \frac{{{\Phi _{{\text{n1}}}}}}{{{\Omega _k}}}}}}$,  ${\text{ }}p_{\text{F}}^{(2)} = 1 - e^{-\frac{\Phi_{11}}{\Omega_2}}$, and  \eqref{eq:dof} holds.
\end{rem}

\section{Dynamic Power Allocation}
\label{sec:DPA}
In this section, the power allocation is adjusted at the relay to ensure signal decoding at user 2. Since it is not tractable to adopt Nakagami-$m$ fading model on the source-relay link due to its mathematical complexity, we use the Rayleigh model in this section. By solving ${\gamma_2} = {\tau_2}$ to find $\delta$ and,  noting that $1- \delta$ must be non-negative, we  obtain
\begin{equation}
\label{eq10}
\delta  = 1 - \max \left\{ {0,\;\frac{{P'g_2 - {\tau_2}\left( {\sigma_e^2P' + 1} \right)}}{{\left( {1 + {\tau_2}} \right)P'g_2}}} \right\}.
\end{equation}

It should be highlighted that the performance of user 2, and consequently user fairness can be further enhanced by letting $\delta$ larger than what is given in (\ref{eq10}). This, however, significantly reduces the performance of user~1 and also the sum rate \cite{vaezi2018myth}.

Based on (\ref{eq5}) and (\ref{eq10}), the outage probability at user~1 for the DPA case is defined as
\begin{equation}
\label{eq11}
p_{\rm{D}}^{\left( 1 \right)} = 1 - \mathbb{P} \left\{ { g > \frac{\tau_0}{P} ,g_2 > {\tau_2}\Big( {\sigma_e^2 + \frac{1}{{P'}}} \Big),{{\gamma}_{\rm{1}}} > {\tau_1}} \right\}.
\end{equation}

\begin{lem}
 For the DPA case, the outage probability of user~1 is given by
\begin{equation}
\begin{array}{l}
p_{\rm{D}}^{\left( 1 \right)}
= 1 -  e^{ - \sigma_e^2\frac{{{\tau_0}}}{{{\Omega_2}}} - \frac{{{\tau_0}d^\alpha }}{{P}}}{\omega_0}{K_1}\left( {{\omega_0}} \right) - \Upsilon_{\rm{II}},
\end{array}
\end{equation}
in which $\Upsilon_{\rm{II}}$ is defined in (\ref{eq23}) and ${\omega_\ell} = 2\sqrt {\frac{1}{{\eta Pd^{ - \alpha }}}\frac{{{\tau_\ell}}}{{{\Omega_2}}}}$, $\ell \in \left\{ 0,2 \right\}$.
\end{lem}
\begin{proof} See Appendix B.
\end{proof}

Further, the outage probability at user~2 for the DPA case is given as
\begin{equation}
\label{eq16}
p_{\rm{D}}^{\left( 2 \right)} = 1 - \mathbb{P} \left\{ {g > \frac{\tau_0}{P},{g_2} > {\tau_2}\Big( {\sigma_e^2 + \frac{1}{{P'}}} \Big)} \right\}.
\end{equation}
Then it is straightforward to show that the outage probability of user 2 is obtained by
\begin{equation}
\label{eq17}
p_{\rm{D}}^{\left( 2 \right)} = 1 - e^{ - \sigma_e^2\frac{{{\tau_2}}}{{{\Omega_2}}} - \frac{{d^\alpha {\tau_{\rm{0}}}}}{{P}}} {\omega_2}{K_1}\left( {{\omega_2}} \right).
\end{equation}
%When $P$ is large, we achieve the asymptotic expression for $p_{\rm{D}}^{\left( 2 \right)}$ as
%\begin{align}
%p^{(2)}_{\rm{D}}
%\to
%	1 - e^{-\sigma^2_e \frac{\tau_2}{\Omega_2}}
%		( 1 + \frac{\omega^2_2}{2} \ln(\omega_2/2))
%\end{align}
\begin{rem}
Similar to the FPA case in Remark~1,  the diversity order is  zero in the DPA case too.
\end{rem}

\section{The Effect of Imperfect SIC}
\label{sec:imp}

 Obtaining  \eqref{eq4} in Section~\ref{sec:model}, we have assumed that SIC
  is perfect. However, since the CSI of the relay-to-user channels is imperfect,
 SIC cannot be error-free in practice.
  In this section,  we investigate  the  detrimental effect of imperfect CSI
 in  SIC and outage probabilities. It should be highlighted that imperfect SIC can significantly
 affect the performance of the system even if the channel estimation is accurate.

To study the  joint impact of imperfect SIC and imperfect CSI,
 an analysis similar to that of \cite{chen2018fully} can be carried out here. However,
 applying this analysis  for the DPA-NOMA scheme is not straightforward and
 requires large modifications which cannot be included due to the space limit. Recall that in the DPA-NOMA case,  the relay does not have the knowledge of the SIC process at user~1 and thus can only adjust the power allocation in a perfect SIC manner. Thus, we only discuss the impact of imperfect SIC to the FPA-NOMA scheme.

We know that the received signal at user 1 (i.e. the stronger user) is given by
\begin{align}
y_1 =& \sqrt{P_r}\left(\sqrt{\delta}{s_2} + \sqrt{1 - \delta} {s_1} \right)(\hat{h}_1 + {e_1}) + {n_1}.
\end{align}
Then, the received signal after the SIC process is given by
\begin{align}
\hat{y}_1 =& \sqrt{P_r(1 - \delta)}{s_1} \hat{h}_1  + \sqrt{P_r \delta} \hat{h}_1 (s_2 - \hat{s}_2) \nonumber \\
	&+ (\sqrt{P_r(1 - \delta)}{s_1} + \sqrt{P_r \delta}  {s_2}){e_1} + {n_1},
\end{align}
where $\hat{s}_2$ is the estimated signal of the user~2 at  user~1. Let $\sigma^2_{ic} \triangleq \mathbb{E}[|s_2 - \hat{s}_2|^2]$ denote the expected residual power level after SIC, $\sigma^2_{ic} = 0$ would refer to perfect SIC case in which $s_2 = \hat{s}_2$.

The SINR at user~1 after imperfect SIC, i.e., the modified version of \eqref{eq3},  becomes
\begin{align}
\label{eqQ5.1}
\gamma_1 = \frac{(1-\delta){P'}{g_1}}{{P'}{\sigma_e^2} + {P'}{\delta} {g_1} {\sigma^2_{ic}} + 1}.
\end{align}
Then, for the FPA-NOMA with imperfect SIC, Lemma~\ref{lem1} requires a slight modification.
Specifically, \eqref{eq:Phi} should be replace by
\begin{align}
\left[
\begin{matrix}
	\Phi_{11} & \Phi_{12} \\
	\Phi_{21} & \Phi_{22} \\
\end{matrix}
\right] \triangleq
\left[
	\frac{\tau_2}{\delta(\tau_2 + 1)-\tau_2} ,
	\frac{\tau_1}{1 - \delta(\tau_1 \sigma^2_{ic} + 1 )}
\right]^T
\left[ \sigma^2_e, \frac{1}{\eta} \right].
\end{align}
Note that $p_{\rm{F}}^{(1)}$ = $p_{\rm{F}}^{(2)}$ = 0 when either $\delta \le \frac{\tau_2}{\tau_2 + 1}$ or $\delta \ge \frac{1}{\tau_1 \sigma^2_{ic} + 1}$.
\begin{proof}
Taking similar steps as in Appendix A, we can get the exact outage probability for this case. Specifically, $p_{\rm{F}}^{(1)}$ and $p_{\rm{F}}^{(2)}$ are simply derived by  substitution into (A-2) which gives
\begin{align}
\Lambda_{\text{max}} &=  \nonumber \\
&\left\{
	\begin{matrix}
		\Phi_{11} + \frac{\Phi_{12}}{Pg-\tau_0}, & \frac{\tau_2}{\tau_2 + 1} < \delta < \min\Big[ \frac{\tau_0 - \tau_1}{\tau_0 + \tau_1 \tau_2 \sigma^2_{ic}},\frac{1}{\tau_1 \sigma^2_{ic} +1} \Big] \\
		\Phi_{21} + \frac{\Phi_{22}}{Pg-\tau_0}, &  \frac{1}{\tau_1 \sigma^2_{ic} +1} > \delta \ge \frac{\tau_0 - \tau_1}{\tau_0 + \tau_1 \tau_2 \sigma^2_{ic}}
	\end{matrix}
\right.
\end{align}
\end{proof}

\begin{figure}
\begin{center}
    \includegraphics[scale=0.55]{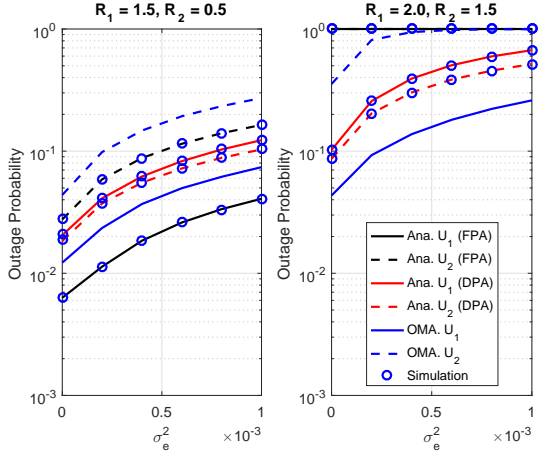}
    \caption{Outage probability  of FPA-NOMA/DPA-NOMA cases versus $\sigma_e^2$, where $ P_s = 15$ dB.}
    \label{fig2}
\end{center}
\end{figure}
\begin{figure}
\begin{center}
    \includegraphics[scale=0.55]{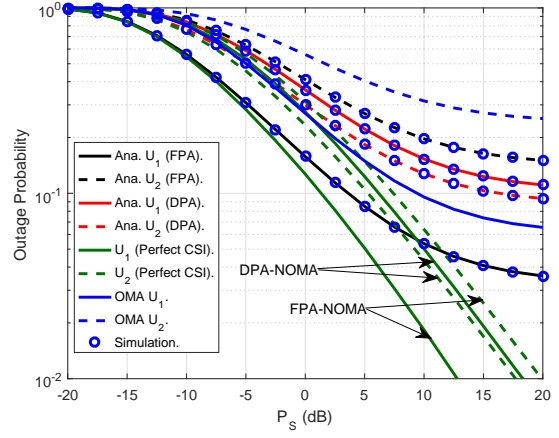}
    \caption{Outage probability of FPA-NOMA/DPA-NOMA case versus $P_s$, where  $ {R_{\rm{2}}} = 0.5, {R_{\rm{1}}} = 1.5$ and $\sigma_e^2 = 0.001.$}
    \label{fig3}
    \end{center}
\end{figure}

\begin{figure}
\begin{center}
    \includegraphics[scale=0.55]{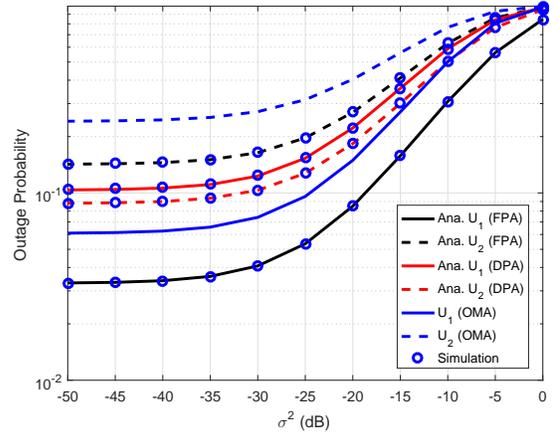}
    \caption{Outage probability of FPA-NOMA/DPA-NOMA for imperfect and
perfect CSI cases versus noise power, where  $ {R_{\rm{2}}} = 0.5, {R_{\rm{1}}} = 1.5$, $\sigma_e^2 =  0.001$ and $ P_s = 15$ dB.}
    \label{fig4}
    \end{center}
\end{figure}

\begin{figure}
\begin{center}
    \includegraphics[scale=0.6]{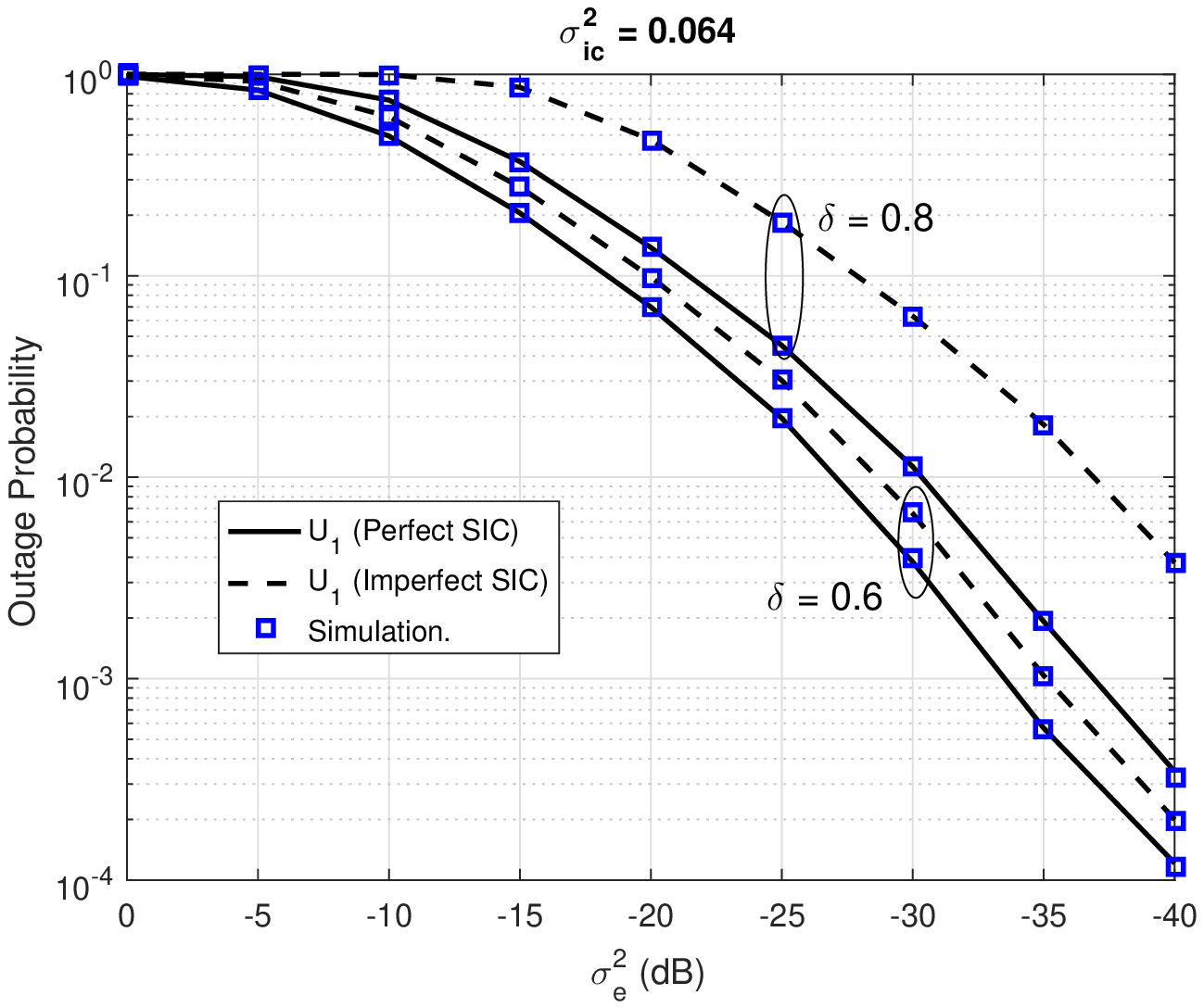}
    \caption{Outage probability of FPA-NOMA case versus $\sigma_e^2$, where  $ {R_{\rm{2}}} = 0.5, {R_{\rm{1}}} = 1.5$, $d = d_1 =1$, $d_2 = 10$ and $P_s = 30$ dB}.
    \label{fig5}
    \end{center}
\end{figure}

\begin{figure}
\begin{center}
    \includegraphics[scale=0.6]{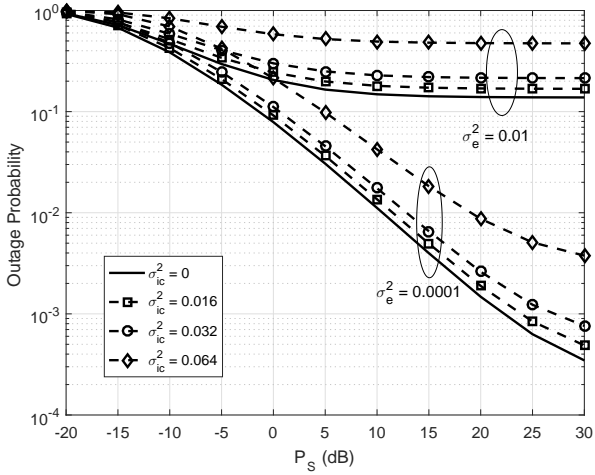}
    \caption{Outage probability of FPA-NOMA case versus $P_s$, where  $ {R_{\rm{2}}} = 0.5, {R_{\rm{1}}} = 1.5$, $d = d_1 =1$ and $d_2 = 10$.}
    \label{fig6}
    \end{center}
\end{figure}

\section{Numerical Results}
\label{sec:sim}

In this section, numerical results are presented to validate the analytic results and to evaluate the outage performance of the downlink EH-NOMA network. The source-relay and relay-users links are considered under Rayleigh fading environment, ${\sigma ^2}$ = -30 dB and $\delta  = 0.8$. The distances between the Relay to each node are ${d} = 1$, ${d_1} = 1$, ${d_2} = 10$ and the path loss exponent is $\alpha  = 2$. Note that the target data rates from Fig.~\ref{fig2} to Fig.~\ref{fig4}, i.e. ${R_1}$ and ${R_2}$, are given in terms of bits/s/Hz.

As can be seen from Fig.~\ref{fig2} to Fig.~\ref{fig4}, the analytical curves match  the simulation curves. In Fig.~\ref{fig2}, the outage probability at user~1 and user~2 increases as $R_1$ and $R_2$ increase. Choosing large values for $R_1$ and $R_2$ could make the outage event always occur, i.e. the outage probabilities become one. Figures~\ref{fig3} and \ref{fig4} confirm that as $P_s$ increases the outage probability decreases. Specifically, in Fig.~\ref{fig3}, FPA-NOMA results in a higher outage at the weak user but provides much higher success probability for the strong user than DPA-NOMA. Unlike that in Fig.~\ref{fig4}, it is seen that FPA-NOMA can provide better performance for both users than OMA. Further, it is observed that an outage floor appears in the outage probability due to the channel estimation errors.

Further, the impact of channel estimation error  $\sigma^2_{e}$ and SIC error $\sigma^2_{ic}$ is investigated in Fig.~\ref{fig5} and Fig.~\ref{fig6}. It is observed that a higher channel estimation accuracy leads to a higher degradation in the outage performance with the same value of $\sigma^2_{ic}$. For example,  for $\sigma^2_e = -40$ dB the outage probability with imperfect SIC is reduced more than 10 times from that with perfect SIC,  but for $\sigma^2_e = -20$ dB this reduction is about 4 times. However, one can improve the system performance as well as the gap between the perfect and imperfect SIC  by increasing the power allocated for user~1 (decreasing $\delta$). In fact, as $\delta$ decreases (but  $\delta > \frac{\tau_2}{\tau_2 + 1}$ holds), the effect of imperfect SIC become less significant, e.g., resulting only 1.7 times reduction in the performance at $\sigma^2_e = -40$ for $\delta = 0.6$. Further, by increasing the transmit power, the outage performance can be improved but reaches different outage floor depending on $\sigma^2_{ic}$ as depicted in Fig.~\ref{fig6}.
\section{Conclusion}
\label{sec:con}
We have examined the outage performance of EH-NOMA networks with imperfect CSI
and imperfect SIC. We have derived closed-form expressions for the outage probability at each user. Simulation results confirm the analytic results and show that NOMA can improve the outage probability of both users. In addition, it is shown that FPA-NOMA achieves better performance for the stronger user than DPA-NOMA but poor fairness. Further, DPA-NOMA provides a higher success probability for the weaker than FPA-NOMA and a better fairness. It is also confirmed that FPA-NOMA can outperform orthogonal multiple access schemes.
\\
\renewcommand{\theequation}{B-\arabic{equation}}
% redefine the command that creates the equation no.
\setcounter{equation}{2}  % reset counter
\begin{figure*}
\begin{align}
\label{eq23}
\Upsilon_{\text{II}}
	&= \sum_{i \in \{1,3\}} \frac{d^\alpha}{\Omega_{\bar{i}} {\eta}{P}}
		e^{ -\frac{\tau_0}{P{d^{-\alpha}}} }
		\int\limits_0^{\tau_1+{\tau_1}{\tau_2}} e^{ -\text{H}_{\bar{i}}(\rho)\sigma_e^2 }
		\int\limits_0^{\infty}\Big(\sigma_e^2 + \frac{1}{\upsilon}\Big) \exp\Big(-\text{H}_{\bar{i}}(\rho)-\frac{\upsilon}{{\eta}{P}{d^{-\alpha}}}\Big) d{\upsilon}d{\rho} \nonumber \\
	&\mathop{=}\limits^{(a)}
		\sum_{i \in \{1,3\}} \frac{1}{\Omega_{\bar{i}}}
		e^{ -\frac{\tau_0}{P{d^{-\alpha}}} }
		\int\limits_0^{\tau_1+{\tau_1}{\tau_2}}
		\mathop{
		\underbrace{		
		e^{ -\text{H}_{\bar{i}}(\rho)\sigma_e^2 }
		\left[ {\sigma_e^2}{2\sqrt{\frac{\text{H}_{\bar{i}}(\rho)}{{\eta}{P}{d^{-\alpha}}}}}
				K_1\left(2\sqrt{\frac{\text{H}_{\bar{i}}(\rho)}{{\eta}{P}{d^{-\alpha}}}}\right)
			  +\frac{2}{{\eta}{P}{d^{-\alpha}}}
			  	K_0\left(2\sqrt{\frac{\text{H}_{\bar{i}}(\rho)}{{\eta}{P}{d^{-\alpha}}}}\right)
		\right]
		}}\limits_{\triangleq h_i(\rho)}d{\rho} \nonumber \\
	&\mathop{\approx}\limits^{(b)}
		\sum_{i \in \{1,3\}} \frac{1}{\Omega_{\bar{i}}}
		e^{ -\frac{\tau_0}{P{d^{-\alpha}}} }
		\frac{\tau_1+{\tau_1}{\tau_2}}{2}
		\sum_{j=1}^J
		\frac{\pi}{j} \Big|\sin\Big(\frac{2j-1}{2J}{\pi}\Big)\Big|{h_i(\rho_n)},
\end{align}
where $\rho_n = \frac{\tau_1+{\tau_1}{\tau_2}}{2}\left[1 + \cos(\frac{2j-1}{2J}{\pi})\right]$ and $J$ is the trade-off coefficient reflecting the approximation accuracy.
\hrule
\end{figure*}

\renewcommand{\theequation}{A-\arabic{equation}}
% redefine the command that creates the equation no.
\setcounter{equation}{0}  % reset counter
\section*{Appendix A}
From (\ref{eq5}) the outage probability of user~1 can be written as
\begin{equation}
\label{eq18}
\begin{array}{l}
p_{\rm{F}}^{( 1 )} = 1 - \mathbb{P} \big \{ g > \frac{\tau_0}P, \; g_1 > \Lambda_{\max }\big\},
\end{array}
\end{equation}
where
\begin{align} \label{eq:Gamma}
\Lambda_{\max }&={\max \left( {{{{\Phi }}_{11}} + \frac{{{{{\Phi }}_{12}}}}{{Pg - {\tau_0}}}, {{{\Phi }}_{21}} + \frac{{{{{\Phi }}_{22}}}}{{Pg - {\tau_0}}}} \right)} \nonumber \\ &= \left\{
                                                            \begin{array}{ll}
                                                              {\Phi}_{11} + \frac{{\Phi}_{12}}{P{g - \tau_0}}, & \frac{\tau_2}{\tau_2 +1} < \delta < \frac{\tau_0 -\tau_1}{\tau_0}\\
                                                               {\Phi}_{21} + \frac{{\Phi}_{22}}{P{g - \tau_0}}, & \qquad \quad \delta \ge  \frac{\tau_0 -\tau_1}{\tau_0}
                                                            \end{array}
                                                          \right.
\end{align}
where the last step can be verified in view of \eqref{eq:Phi}.
Since ${F_{g_1}}\left( \gamma  \right) = 1 - \sum\nolimits_{k = 1}^3 {{{\left( { - 1} \right)}^{k - 1}}\exp ( { - \frac{\gamma }{{{\Omega_k}}}} )}$  is the CDF of $g_1$ and ${f_{g}}\left( \gamma  \right) = \frac{{{m^m}d^{\alpha m}}}{{\Gamma \left( m \right)}}{\gamma ^{m - 1}}\exp ( { - md^\alpha \gamma })$ is the PDF of ${g}$ \cite{yang2017impact},  $p_{\rm{F}}^{\left( 1 \right)}$ can be expressed as
\begin{align}
\label{eq19}
%\begin{array}{c}
p_{\rm{F}}^{\left( 1 \right)} = 1 &- \sum\nolimits_{k = 1}^3 {{{\left( { - 1} \right)}^{k - 1}}\frac{{{m^m}}}{{\Gamma \left( m \right)d^{ - \alpha m}}}} \notag \\
 &\times \underbrace {\int_{{{{\tau_0}} \mathord{\left/
 {\vphantom {{{\tau_0}} {P}}} \right.
 \kern-\nulldelimiterspace} {P}}}^\infty  {{z^{m - 1}}\exp \left( { - \frac{{{\Lambda_{\max }}}}{{{\Omega_k}}} - \frac{{mz}}{{d^{ - \alpha }}}} \right)dz} }_{{{\rm X}_m}}.
%\end{array}
\end{align}

Next, using  \eqref{eq:Gamma} for $\frac{\tau_2}{\tau_2 +1} < \delta < \frac{\tau_0 -\tau_1}{\tau_0}$, with the help of \cite[Eq. 3.471.9]{jeffrey2007table}, the integral  in (\ref{eq19}) is evaluated as
\begin{align}
\label{eq20}
%\begin{array}{c}
{{\rm X}_m} &= \sum\limits_{i = 0}^{m - 1} {{\begin{array}{*{20}{c}}
{m-1}\choose{i}
\end{array}} e^{ - \frac{{m{\tau_0}}}{{Pd^{ - \alpha }}}} {{\left( {\frac{{{\tau_0}}}{{P}}} \right)}^{m - i - 1}}}  \notag\\
& \times 2 e^{ - \frac{{{{{\Phi }}_{11}}}}{{{\Omega_k}}}}{\left( {\sqrt {\frac{{{{{\Phi }}_{12}}}}{{{\Omega_k}}}\frac{{d^{ - \alpha }}}{{mP}}} } \right)^{i + 1}}{K_{i + 1}}\left( {{\beta_{1,k}}} \right),
%\end{array}
\end{align}
in which  $ {n}\choose{k}$  shows the binomial coefficient. Then, it can be checked that
for $\delta \ge \frac{\tau_0 -\tau_1}{\tau_0}$, the value of
${\rm X}_m$ is obtained by replacing ${{{\Phi }}_{11}}$ and ${{{\Phi }}_{12}}$  with ${{{\Phi }}_{21}}$ and ${{{\Phi }}_{22}}$ in (\ref{eq20}), respectively.
Substituting (\ref{eq20}) into (\ref{eq19}), respecting the value of $\delta $,
after some mathematical manipulations the proof is completed.

\renewcommand{\theequation}{B-\arabic{equation}}
% redefine the command that creates the equation no.
\setcounter{equation}{0}  % reset counter
\section*{Appendix B}
From (\ref{eq10}) the outage probability of user~1 in (\ref{eq11}) can be further expressed as
\begin{equation}
\label{eq21}
\begin{array}{c}
p_{\rm{F}}^{\left( 1 \right)} = 1 - \mathbb{P} \left\{ {g_1 > g_2 > {\tau_0}\left( {\frac{1}{{P'}} + \sigma_e^2} \right),g > \frac{{{\tau_0}}}{{P}}} \right\}\\
 - \mathbb{P} \left\{ {g_1 > \frac{{{\tau_1}\left( {P'\sigma_e^2 + 1} \right)\left( {1 + {\tau_2}} \right)g_2}}{{P'g_2 - \left( {P'\sigma_e^2 + 1} \right){\tau_2}}},g > \frac{{{\tau_0}}}{{P}},} \right.\\
{\rm{       }}\left. {{\rm{   }}{\tau_2}\left( {\frac{1}{{P'}} + \sigma_e^2} \right) < g_2 < {\tau_0}\left( {\frac{1}{{P'}} + \sigma_e^2} \right)} \right\}.
\end{array}
\end{equation}
Note that ${f_{g_2,g_1}}\left( {x,y} \right) = \sum\limits_{i \in \left\{ {1,3} \right\}} {\frac{1}{{{\Omega_{\bar i}}{\Omega_i}}}\exp \left( { - \frac{x}{{{\Omega_{\bar i}}}} - \frac{y}{{{\Omega_i}}}} \right)}$ and ${\tau_1} + {\tau_2} + {\tau_1}{\tau_2} = {\tau_0}$. The proof of the first probability is not shown in this paper as it involves simpler mathematical computations than the second one and can be easily derived by adopting steps in Appendix A. The second probability in (\ref{eq21}), denoted as ${\Upsilon_{{\rm{II}}}}$, is obtained as follows
\begin{equation}
\label{eq22}
{\Upsilon_{{\rm{II}}}} = \int\limits_{\frac{{{\tau_0}}}{{P}}}^\infty  {\int\limits_{{\tau_2}\left( {\sigma_e^2 + \frac{1}{{\eta \left( {Pz - {\tau_{\rm{0}}}} \right)}}} \right)}^{{\tau_0}\left( {\sigma_e^2 + \frac{1}{{\eta \left( {Pz - {\tau_{\rm{0}}}} \right)}}} \right)} {\int\limits_{\nu \left( {x,z} \right)}^\infty  {{f_{g_2,g_1}}\left( {x,y} \right){f_{g}}\left( z \right)dxdydz} } },
\end{equation}
where $\nu \left( {x,z} \right) \triangleq \frac{{{\tau_1}\left( {1 + {\tau_2}} \right)\left( {\eta \left( {Pz - {\tau_{\rm{1}}}} \right)\sigma_e^2 + 1} \right)x}}{{\eta \left( {Pz - {\tau_{\rm{0}}}} \right)x - {\tau_2}\left( {\eta \left( {Pz - {\tau_{\rm{0}}}} \right)\sigma_e^2 + 1} \right)}}$.

By letting $\upsilon  = \eta \left( {Pz - {\tau_0}} \right)$ and then $\rho  = \frac{\upsilon }{{\upsilon \sigma_e^2 + 1}}x - {\tau_2}$, the above equality can be expressed by (B.3),
%% Equation (B.3)
where ${{\rm H}_i}\left( \rho  \right) \triangleq \left( {\frac{{{\tau_{\rm{1}}}\left( {1 + {\tau_{\rm{2}}}} \right)}}{\rho }\frac{1}{{{\Omega_{\bar i}}}} + \frac{1}{{{\Omega_i}}}} \right)\left( {\rho  + {\tau_{\rm{2}}}} \right)$, $\bar i = 3/i$, (a) is due to \cite[Eq.3.471.9]{jeffrey2007table} and (b) is obtained by using Gaussian-Chebyshev quadrature \cite{yang2017impact}.


\begin{thebibliography}{10}

\bibitem{NOMAbook}
M.~Vaezi, Z.~Ding, and H.~V. Poor, {\em {Multiple Access Techniques for 5G
  Wireless Networks and Beyond}}.
\newblock Springer, 2019.

\bibitem{saito2013non}
Y.~Saito, Y.~Kishiyama, A.~Benjebbour, T.~Nakamura, A.~Li, and K.~Higuchi,
  ``{Non-orthogonal multiple access (NOMA) for cellular future radio access},''
  in {\em Proc. IEEE 77th VTC Spring}, pp.~1--5, 2013.

\bibitem{ding2014performance}
Z.~Ding, Z.~Yang, P.~Fan, and H.~V. Poor, ``{On the performance of
  non-orthogonal multiple access in 5G systems with randomly deployed users},''
  {\em IEEE Signal Proc. Lett.}, vol.~21, no.~12, pp.~1501--1505, 2014.

\bibitem{shin2017non}
W.~Shin, M.~Vaezi, B.~Lee, D.~J. Love, J.~Lee, and H.~V. Poor, ``Non-orthogonal
  multiple access in multi-cell networks: Theory, performance, and practical
  challenges,'' {\em IEEE Commun. Mag.}, vol.~55, no.~10, pp.~176--183, 2017.

\bibitem{yeu2018exploiting}
X.~Yue, Y.~Liu, S.~Kang, A.~Nallanathan, and Z.~Ding, ``Exploiting
  full/half-duplex user relaying in noma systems,'' {\em IEEE Trans. Commun.},
  vol.~66, no.~2, pp.~560--575, 2018.

\bibitem{pin2013ambient}
M.~Pi{\~n}uela, P.~D. Mitcheson, and S.~Lucyszyn, ``{Ambient RF energy
  harvesting in urban and semiurban environments},'' {\em IEEE Trans. Microw.
  Theory Techn.}, vol.~61, pp.~2715--2726, 2013.

\bibitem{val2014har}
C.~Valenta and G.~Durgin, ``Harvesting wireless power: Survey of
  energy-harvester conversion efficiency in far-field, wireless power transfer
  systems,'' {\em IEEE Microw. Mag.}, vol.~15, pp.~108--120, 2014.

\bibitem{sbi2015wirel}
S.~Bi, C.~K. Ho, and R.~Zhang, ``Wireless powered communication: opportunities
  and challenges,'' {\em IEEE Commun. Mag.}, vol.~53, pp.~117--125, 2015.

\bibitem{rzhang2013mimo}
R.~Zhang and C.~K. Ho, ``{MIMO broadcasting for simultaneous wireless
  information and power transfer},'' {\em IEEE Trans. Wireless Commun.},
  vol.~12, pp.~1989--2001, 2013.

\bibitem{lliu2013wirel}
L.~Liu, R.~Zhang, and K.-C. Chua, ``Wireless information transfer with
  opportunistic energy harvesting,'' {\em IEEE Trans. Wireless Commun.},
  vol.~12, pp.~288--300, 2013.

\bibitem{jpark2013joint}
J.~Park and B.~Clerckx, ``{Joint wireless information and energy transfer in a
  two-user MIMO interference channel},'' {\em IEEE Trans. Wireless Commun.},
  vol.~12, pp.~4210--4221, 2013.

\bibitem{ITUSG05}
ITU, ``{Minimum requirements related to technical performance for IMT-2020
  radio interface(s)}.'' \url{https://www.itu.int/md/R15-SG05-C-0040/en}, 2017.
\newblock [Online; accessed August 2008].

\bibitem{liu2016cooperative}
Y.~Liu, Z.~Ding, M.~Elkashlan, and H.~V. Poor, ``{Cooperative non-orthogonal
  multiple access with simultaneous wireless information and power transfer},''
  {\em IEEE J. Sel. Areas Commun.}, vol.~34, no.~4, pp.~938--953, 2016.

\bibitem{han2016performance}
W.~Han, J.~Ge, and J.~Men, ``{Performance analysis for NOMA energy harvesting
  relaying networks with transmit antenna selection and maximal-ratio combining
  over Nakagami-m fading},'' {\em IET Commun.}, vol.~10, no.~18,
  pp.~2687--2693, 2016.

\bibitem{nguyen2017maximum}
X.-X. Nguyen and D.-T. Do, ``{Maximum harvested energy policy in full-duplex
  relaying networks with SWIPT},'' {\em Int. J. Commun. Sys.}, no.~7,
  pp.~1--16, 2017.

\bibitem{do2017bnbf}
N.~T. Do, D.~B. Da~Costa, T.~Q. Duong, and B.~An, ``{A BNBF user selection
  scheme for NOMA-based cooperative relaying systems with SWIPT},'' {\em IEEE
  Commun. Lett.}, vol.~21, no.~3, pp.~664--667, 2017.

\bibitem{yang2017impact}
Z.~Yang, Z.~Ding, P.~Fan, and N.~Al-Dhahir, ``{The impact of power allocation
  on cooperative non-orthogonal multiple access networks with SWIPT},'' {\em
  IEEE Trans. Wireless Commun.}, vol.~16, no.~7, pp.~4332--4343, 2017.

\bibitem{vaezi2018myth}
M.~Vaezi, R.~Schober, Z.~Ding, and H.~V. Poor, ``{Non-orthogonal multiple
  access: Common myths and critical questions},''
\newblock \url{https://arxiv.org/abs/1809.07224}.

\bibitem{do2018improving}
T.~N. Do, D.~B. da~Costa, T.~Q. Duong, and B.~An, ``Improving the performance
  of cell-edge users in noma systems using cooperative relaying,'' {\em IEEE
  Trans. Commun.}, vol.~66, no.~5, pp.~1883--1901, 2018.

\bibitem{jeffrey2007table}
A.~Jeffrey and D.~Zwillinger, {\em Table of integrals, series, and products}.
\newblock Academic Press, 2007.

\bibitem{chen2018fully}
X.~Chen, Z.~Zhang, C.~Zhong, R.~Jia, and D.~W.~K. Ng, ``Fully non-orthogonal
  communication for massive access,'' {\em IEEE Trans. Commun.}, vol.~66,
  no.~4, pp.~1717--1731, 2018.

\end{thebibliography}
\end{document}